\documentclass[12pt,a4paper,titlepage]{article}

\usepackage{indentfirst} %% Отсуп у первой строки раздела
\usepackage{lastpage} %% Количество страниц

\usepackage{amsmath} 
\usepackage{amssymb}
\usepackage{amsfonts}
\usepackage{amsthm}
\usepackage{bm} % Делает любые символы жирнными
\usepackage{icomma} % "Умная" запятая: $0,2$ --- число, $0, 2$ --- перечисление

\usepackage{extsizes} % Правильно меняет размер тескта
\usepackage{geometry} % Поля
\usepackage{setspace}
\usepackage{textcase}

\usepackage{cancel}

\usepackage[linkcolor=blue,colorlinks=true]{hyperref}

\newtheorem{theorem}{Theorem}
\newtheorem{lemma}{Lemma}

\theoremstyle{definition}
\usepackage{epsfig}
\usepackage{graphicx}

\begin{document}

	\begin{center}
		\Large
	\textbf{Constructing Fermionic Dynamics with Closed Moment Hierarchies}
	
		\large 
		\textbf{A.E. Teretenkov}\footnote{Department of Mathematical Methods for Quantum Technologies, Steklov Mathematical Institute of Russian Academy of Sciences,
			ul. Gubkina 8, Moscow 119991, Russia\\ E-mail:\href{mailto:taemsu@mail.ru}{taemsu@mail.ru}}
		\end{center}

	\begin{center}
		\begin{minipage}{0.9\linewidth}
			
			We construct a broad class of completely positive maps and Go\-rini--Kossakowski--Sudarshan-Lindblad  generators for fermionic systems induced by linear transformations of system and environment  modes. For these maps, we derive explicit Heisenberg-picture formulas for arbitrary normally ordered monomials in terms of minors of the underlying mode-transformation matrices and environment correlation tensors. We show that for even environment states the linear span of monomials up to any fixed order is invariant, which yields closed equations for low-order moments and makes their computation efficient. We also discuss the relation of this construction to second quantization of non-Hermitian one-particle contractions and extend the formalism to completely positive maps arising from post-selection.
		\end{minipage}
	\end{center}

   \section{Introduction}
   
   It is well known that, in both bosonic and fermionic settings, quantum Gaussian channels and the corresponding one-parameter semigroups possess a remarkable structural property: under their action, products of creation and annihilation operators are transformed into linear combinations of operators of the same or lower order \cite{DodMan86,Bravyi2005,Prosen2008,Prosen2010,Heinosaari2010,Prosen2010b,Ter19,Ter19R,Nosal2022}. This makes Gaussian models exceptionally tractable, since low-order correlation functions evolve in closed systems of equations and can often be computed without dealing with the full many-body density matrix. For this reason, quantum Gaussian semigroups still find new applications \cite{Bakker2020,Flynn2021, Gaidash2025, Shustin2025, Bhattacharyya2026} and new directions for analysis \cite{Fagnola2022, Agredo2022, Fagnola2024, Fagnola2025, Fang2026}. 
   
   At the same time, Gaussianity is not the only mechanism leading to closed dynamics of moments \cite{Vanheuverzwijn1978, Demoen1979, Holevo1996, Holevo98}. In recent years, several non-Gaussian models with similar invariant-space properties were identified and actively studied \cite{Teretenkov20,NosTer20,Barthel2022,Ivanov2022, Zhang2022, Teretenkov2024a,Barchielli2024, Penc2026,Teretenkov2026, Biriukov2026}. These examples indicate that the class of open quantum evolutions admitting closed equations for low-order moments is substantially wider than the Gaussian one. This raises a natural structural question: how can one systematically construct completely positive maps and semigroups whose action on fermionic monomials remains explicitly controllable, and for which moments up to a prescribed order evolve independently of higher-order ones?
   
   In the present work we address this question in the fermionic setting. We construct a broad class of completely positive maps generated by linear transformations of canonical anticommutation relation (CAR) operators on an enlarged system--environment space. In this sense, the present work can be viewed as a multimode fermionic analogue of our previous work \cite{Teretenkov2026}, where a similar idea was used for a single bosonic mode. The construction is based on a unitary mixing of system and environment fermionic modes together with a partial trace over the environment. Working in a concrete tensor-product realization of two mutually anticommuting CAR families, we derive an explicit formula for the Heisenberg-picture action of the resulting map on arbitrary normally ordered monomials $f_J^\dagger f_I$.  This formula is expressed through minors of the matrices defining the linear transformation and through the environment correlation tensors. In this way, the action of the channel is reduced to finite-dimensional algebraic data. Let us also remark that this formula can be viewed as a fermionic counterpart of formulas recently derived in the bosonic sampling context in \cite{Biriukov2026}.
   
   Our first main observation discussed in Section~\ref{sec:fermChann} is that for general environment states the obtained formula is explicit but still involves parity contributions, which in principle couple low-order moments to monomials of arbitrarily high order. However, this obstruction disappears for even environment states. In that case, the parity terms vanish and the action of the channel preserves the linear span of normally ordered monomials up to any fixed total order. As a result, all moments up to a fixed  order evolve in a closed way and can be computed without reference to higher-order correlations. This gives a wide family of fermionic channels with the same kind of computationally favorable closure property that is usually associated with Gaussian dynamics. We then analyze several important special cases. 
   
   A separate motivation of the paper is the problem of second quantization of non-Hermitian one-particle dynamics discussed in Section~\ref{sec:secondQuant}. We show that every contraction matrix can be realized as the linear Heisenberg action of a completely positive fermionic map on annihilation operators. This provides a large family of possible many-body extensions of dissipative one-particle evolutions. In general such extensions are highly non-unique, and for arbitrary environment states they do not lead to a genuine multiparticle semigroup on all monomials. Nevertheless, in the gauge-invariant case the action on products of annihilation operators is governed by exterior powers of the same one-particle contraction, which makes the construction a natural fermionic analogue of second quantization for non-Hermitian effective Hamiltonians.
   
   In Section~\ref{sec:dynamics} we also show that the same class of maps can be used to generate genuine Markovian dynamics. Namely, each completely positive map of the constructed form gives rise to a  Gorini--Kossakowski--Sudarshan--Lindblad  (GKSL) generator \cite{Gorini1976,Lindblad1976}, whose Heisenberg action on normally ordered monomials is again explicit. This immediately implies that moments up to any fixed order satisfy a closed linear system of ordinary differential equations. Hence, instead of solving the full master equation on an exponentially large state space, one can compute low-order observables by exponentiating matrices whose size grows only polynomially with the number of modes. This provides an efficient framework for studying low-order moment dynamics.
   
   In Section~\ref{sec:posSec} we extend the construction from channels to more general completely positive maps arising from post-selection on measurement outcomes in the environment. In this case the dynamics is described by inserting an effect operator into the Stinespring-type construction, and we again derive an explicit Heisenberg-picture formula for the action on fermionic monomials. For even environment states the parity obstruction again disappears. Thus, the developed formalism naturally covers not only trace-preserving evolutions, but also conditional maps and quantum-instrument-type transformations relevant for selective measurements and post-selected open system dynamics.

   \section{Fermionic channels induced by linear mode transformations}
   \label{sec:fermChann}
   
   First, we introduce the notation used throughout the paper. Let $a_1,\dots,a_m$ and $b_1,\dots,b_m$ be families of operators on $\mathbb{C}^{2^m} \otimes \mathbb{C}^{2^m}$ such that each $a_i$ and $b_j$ is a fermionic annihilation operator. 
   More precisely, we consider  the operators $a_i,a_i^\dagger$ and $b_j,b_j^\dagger$ satisfying CAR inside each family
   \begin{equation*}
   	\{a_i,a_j^\dagger\} = \delta_{ij}\,\mathbf 1,
   	\qquad
   	\{a_i,a_j\} = 0,
   	\qquad
   	\{a_i^\dagger,a_j^\dagger\} = 0,
   \end{equation*}
   \begin{equation*}
   	\{b_i,b_j^\dagger\} = \delta_{ij}\,\mathbf 1,
   	\qquad
   	\{b_i,b_j\} = 0,
   	\qquad
   	\{b_i^\dagger,b_j^\dagger\} = 0,
   \end{equation*}
   and, in addition, the two families mutually anticommute:
   \begin{equation*}
   	\{a_i,b_j\} = 0,
   	\qquad
   	\{a_i,b_j^\dagger\} = 0,
   	\qquad
   	\{a_i^\dagger,b_j\} = 0,
   	\qquad
   	\{a_i^\dagger,b_j^\dagger\} = 0,
   \end{equation*}
   for all $i,j=1,\dots,m$, where $\{X,Y\}=XY+YX$. Equivalently, we consider the total family  $a_1,\dots,a_m, b_1,\dots,b_m$ satisfying CAR.
   
   In what follows, we fix a concrete (parity trick) representation on $\mathbb{C}^{2^m} \otimes \mathbb{C}^{2^m}$, where $f_1,\dots,f_m$ are fermionic annihilation operators on $\mathbb{C}^{2^m}$, 
   \begin{equation*}
   	N_f \equiv \sum_{j=1}^m f_j^\dagger f_j,
   	\qquad
   	P = (-1)^{N_f},
   \end{equation*}
   and define
   \begin{equation}\label{eq:parityTrick}
   	a_i \equiv f_i \otimes \mathbf 1,
   	\qquad
   	b_i \equiv P \otimes f_i,
   	\qquad i=1,\dots,m.
   \end{equation}
   Most of our results are formulated in this specific representation. On the other hand, we do not need to fix a  concrete representation of $f_j$ on $\mathbb{C}^{2^m}$. If an explicit realization is needed, one can use the inverse Jordan--Wigner transformation to build up $f_j$ based on tensor products of Pauli matrices.
   
   Introduce column vectors
   \begin{equation*}
   	a = (a_1,\dots,a_m)^{\mathsf T}, 
   	\qquad 
   	b = (b_1,\dots,b_m)^{\mathsf T}.
   \end{equation*}
   Similarly we write the parity trick representation \eqref{eq:parityTrick} as
   \begin{equation}
   	a = f \otimes \mathbf 1,
   	\qquad
   	b = P \otimes f,
   \end{equation}
   where $f = (f_1,\dots,f_m)^{\mathsf T}$.
   For matrices $A,B\in\mathbb{C}^{m\times m}$, define linear combinations componentwise by
   \begin{equation*}
   	(Aa)_i = \sum_{k=1}^m A_{ik} a_k,
   	\qquad
   	(Bb)_i = \sum_{\alpha=1}^m B_{i\alpha} b_\alpha,
   \end{equation*}
   and set $(Aa+Bb)_i = (Aa)_i + (Bb)_i$.
   
   For ordered multi-index $L=(l_1<\dots<l_q)$ and $\Lambda=(\lambda_1<\dots<\lambda_r)$, we write
   \begin{equation*}
   	a_L \equiv a_{l_1}\cdots a_{l_q},
   	\qquad
   	b_\Lambda \equiv b_{\lambda_1}\cdots b_{\lambda_r}.
   \end{equation*}
   We also define
   \begin{equation*}
   	a_L^{\dagger} \equiv (a_L)^{\dagger}  =  a_{l_q}^{\dagger}\cdots a_{l_1}^{\dagger}
   \end{equation*}
   with reverse order of $ a_{l_j}$, which is consistent with hermitian conjugation.

   \begin{lemma}\label{lem:femionicMultinonmial}	
   	Let $A,B\in\mathbb{C}^{m\times m}$ and fix an ordered multi-index $I=(i_1<\dots<i_p)$.
   	Then
   	\begin{equation}
   		(A a + B b)_{I}
   		=
   		\sum_{|I|= |L|+|\Lambda|}
   		\det\!\bigl( A_{I\times L}\,|\,B_{I\times \Lambda} \bigr)\;
   		a_L\, b_\Lambda,
   		\label{eq:compact-noT}
   	\end{equation}
   	where 	
   	\begin{equation*}
   		(A a + B b)_{I} \equiv \prod_{s=1}^p (A a + B b)_{i_s},
   	\end{equation*}
   	the sum runs over all ordered subsets $	L=(l_1<\dots<l_{p-r})$ and $\Lambda=(\lambda_1<\dots<\lambda_r)$, such that $|I|= |L|+|\Lambda|$, we denote by $A_{I\times L}$ (resp.\ $B_{I\times \Lambda}$) the submatrix of $A$
   	(resp.\ $B$) formed by rows $I$ and columns $L$ (resp.\ $\Lambda$), and by
   	\[
   	(A_{I\times L}\,|\,B_{I\times \Lambda}) = 
   	\begin{pmatrix}
   		A_{i_1 l_1} & \dots & A_{i_1 l_{p-r}} &
   		B_{i_1 \lambda_1} & \dots & B_{i_1 \lambda_r}\\
   		\vdots & & \vdots & \vdots & & \vdots\\
   		A_{i_p l_1} & \dots & A_{i_p l_{p-r}} &
   		B_{i_p \lambda_1} & \dots & B_{i_p \lambda_r}
   	\end{pmatrix} \in  \mathbb{C}^{p\times p}
   	\]
   	their horizontal concatenation.
   \end{lemma}

   \begin{proof}
   	For brevity, denote
   	\[
   	x_s := (Aa+Bb)_{i_s}
   	=
   	\sum_{k=1}^m A_{i_s k} a_k
   	+
   	\sum_{\alpha=1}^m B_{i_s \alpha} b_\alpha,
   	\qquad s=1,\dots,p.
   	\]
   	Then
   	\[
   	(Aa+Bb)_I = x_1\cdots x_p.
   	\]
   	
   	Since all operators $a_k$ and $b_\alpha$ mutually anticommute, we have
   	\[
   	x_s x_t = - x_t x_s,
   	\qquad s\neq t.
   	\]
   	Hence the product $x_1\cdots x_p$ is multilinear and totally antisymmetric
   	with respect to the choice of one summand from each factor $x_s$.
   	
   	Now expand the product
   	\[
   	x_1\cdots x_p
   	=
   	\prod_{s=1}^p
   	\left(
   	\sum_{k=1}^m A_{i_s k} a_k
   	+
   	\sum_{\alpha=1}^m B_{i_s \alpha} b_\alpha
   	\right).
   	\]
   	Each summand in this expansion is obtained by choosing, for every $s=1,\dots,p$,
   	either one term $A_{i_s k}a_k$ or one term $B_{i_s\alpha}b_\alpha$.
   	
   	Fix $r\in\{0,\dots,p\}$, and consider those terms for which exactly $r$ factors
   	are chosen from the $Bb$ part and the remaining $p-r$ factors from the $Aa$ part. We collect all terms which,
   	after reordering, give the monomial
   	\[
   	a_L b_\Lambda
   	=
   	a_{l_1}\cdots a_{l_{p-r}}\, b_{\lambda_1}\cdots b_{\lambda_r}.
   	\]
   	
   	Such terms are in one-to-one correspondence with permutations of the $p$ columns
   	of the block matrix $(A_{I\times L}\,|\,B_{I\times \Lambda})$. 
   	Indeed, choosing in the $s$-th factor one of the operators
   	\[
   	a_{l_1},\dots,a_{l_{p-r}}, b_{\lambda_1},\dots,b_{\lambda_r}
   	\]
   	amounts to choosing one column of this matrix in the $s$-th row. Since each of
   	the operators in $a_L b_\Lambda$ must appear exactly once, the admissible choices
   	are parametrized by permutations $\pi\in S_p$.
   	
   	For a fixed permutation $\pi$, the corresponding contribution equals
   	\[
   	\prod_{s=1}^p M_{s,\pi(s)}
   	\]
   	times the operator monomial obtained from $c_{\pi(1)}\cdots c_{\pi(p)}$,
   	where
   	\[
   	(c_1,\dots,c_p)
   	=
   	(a_{l_1},\dots,a_{l_{p-r}}, b_{\lambda_1},\dots,b_{\lambda_r}).
   	\]
   	Reordering this product into the canonical order
   	\[
   	a_{l_1}\cdots a_{l_{p-r}}\, b_{\lambda_1}\cdots b_{\lambda_r}
   	\]
   	produces the sign $\operatorname{sgn}(\pi)$, because all generators mutually
   	anticommute. Therefore the total coefficient of $a_L b_\Lambda$ is
   	\[
   	\sum_{\pi\in S_p}
   	\operatorname{sgn}(\pi)
   	\prod_{s=1}^p M_{s,\pi(s)},
   	\]
   	where $	M=(A_{I\times L}\,|\,B_{I\times \Lambda})$.	By the Leibniz formula for the determinant, this is exactly
   	\[
   	\det\!\bigl(A_{I\times L}\,|\,B_{I\times \Lambda}\bigr).
   	\]
   	
   	Summing over all possible values of $r$, and over all ordered sets
   	\[
   	L\subset\{1,\dots,m\},\quad |L|=p-r,
   	\qquad
   	\Lambda\subset\{1,\dots,m\},\quad |\Lambda|=r,
   	\]
   	we obtain
   	\begin{equation*}
   		(A a + B b)_{I}
   		=
   		\sum_{r=0}^p
   		\sum_{\substack{
   				L\subset\{1,\dots,m\},\,|L|=p-r\\
   				\Lambda\subset\{1,\dots,m\},\,|\Lambda|=r
   		}}
   		\det\!\bigl( A_{I\times L}\,|\,B_{I\times \Lambda} \bigr)\;
   		a_L\, b_\Lambda,
   	\end{equation*}
   	which is exactly the sum in Eq.~\eqref{eq:compact-noT}.
   \end{proof}

   Let $A,B\in\mathbb C^{m\times m}$ satisfy
   \begin{equation}
   	A A^\dagger + B B^\dagger = \mathbf 1
   	\label{eq:isometry-condition-AB}
   \end{equation}
   Introduce the matrix
   \[
   R \equiv (A\,|\,B)\in\mathbb C^{m\times 2m}.
   \]
   Then condition \eqref{eq:isometry-condition-AB} is equivalent to
   \[
   R R^\dagger=\mathbf 1_m,
   \]
   which means that the rows of $R$ form an orthonormal system in $\mathbb C^{2m}$.
   Hence $R$ defines an isometric embedding $\mathbb C^{m}\hookrightarrow \mathbb C^{2m}$. Any isometry can be extended to a unitary
   operator. Therefore, there exist matrices $C,D\in\mathbb C^{m\times m}$ such that
   \begin{equation*}
   	W=
   	\begin{pmatrix}
   		A & B\\
   		C & D
   	\end{pmatrix}
   	\in U(2m).
   	\label{eq:unitary-completion-W}
   \end{equation*}
   The unitary matrix $W$ defines a linear transformation of the fermionic modes,
   \begin{equation*}
   	c=
   	\begin{pmatrix}
   		a\\ b
   	\end{pmatrix},
   	\qquad
   	c' = W c,	
   \end{equation*}
   which preserves the canonical anticommutation relations. Therefore, by
   uniqueness of the finite-dimensional irreducible representation of the CAR
   algebra up to unitary equivalence, this transformation is unitarily implementable in the chosen representation, i.e., there exists
   a unitary operator $U: \mathbb{C}^{2^m} \otimes \mathbb{C}^{2^m} \rightarrow \mathbb{C}^{2^m} \otimes \mathbb{C}^{2^m}$ such that
   \begin{equation*}
   	U^\dagger c\, U = W c.
   \end{equation*}
   In particular,
   \begin{equation*}
   	U^\dagger a U = A a + B b.
   \end{equation*}
   This unitary realization allows us to construct a completely positive map on $\mathbb{C}^{2^m \times 2^m}$ by tracing out the second tensor factor in $ \mathbb{C}^{2^m} \otimes \mathbb{C}^{2^m}$, for which action in the Heisenberg picture  on arbitrary ordered monomials $f_K^\dagger f_L$ is known explicitly.
   
   \begin{lemma}\label{lem:Phi-general-compact-parity}
   	Assume   that $U: \mathbb{C}^{2^m} \otimes \mathbb{C}^{2^m} \rightarrow \mathbb{C}^{2^m} \otimes \mathbb{C}^{2^m}$  is a unitary operator, which acts as
   	\begin{equation}\label{eq:linearTransform}
   		U^\dagger a U = A a + B b
   	\end{equation}
   	for some $A,B\in\mathbb C^{m\times m}$ such that
   	\begin{equation*}
   		A A^\dagger + B B^\dagger = \mathbf 1.
   	\end{equation*}
   	Let $\sigma: \mathbb{C}^{2^m}  \rightarrow \mathbb{C}^{2^m}$, $\sigma = \sigma^{\dagger}$, $\sigma  \ge 0$ act on the second tensor factor, and define a completely positive map
   	\begin{equation}\label{eq:defOfCPmap}
   		\Phi(\rho)
   		\equiv
   		\operatorname{Tr}_2\!\bigl[
   		U(\rho \otimes \sigma)U^\dagger
   		\bigr].
   	\end{equation}	
   	Let
   	$
   	I=(i_1<\dots<i_p)$,
   	$
   	J=(j_1<\dots<j_q)
   	$
   	be ordered multi-indices. 
   	
   	Then
   	\begin{equation}
   		\Phi^*\!\bigl(f_J^\dagger f_I\bigr)
   		=
   		\sum_{\substack{
   				|J|=|K|+|\Xi|\\
   				|I|=|L|+|\Omega|
   		}}
   		(-1)^{|\Xi|(|K|+|L|)}
   		\det\!\bigl( A_{J\times K}\,|\,B_{J\times \Xi} \bigr)^*\,
   		\Gamma_{\Xi;\Omega}\,
   		\det\!\bigl( A_{I\times L}\,|\,B_{I\times \Omega} \bigr)
   		f_K^\dagger f_L \, P^{|\Xi|+|\Omega|},
   		\label{eq:Phi-general-compact-parity}
   	\end{equation}
   	where the sum runs over all subsets $
   	K,L,\Xi,\Omega \subseteq \{1,\dots,m\}
   	$
   	such that
   	$
   	|K|+|\Xi|=|J|,
   	$
   	$
   	|L|+|\Omega|=|I|,
   	$
   	and
   	\begin{equation}\label{eq:envMoments}
   		\Gamma_{\Xi;\Omega}
   		\equiv
   		\operatorname{Tr}\!\bigl(
   		\sigma\, f_\Xi^\dagger \, f_\Omega
   		\bigr).
   	\end{equation}
   \end{lemma}
   
   \begin{proof}
   	By definition of the dual map,
   	\begin{equation}\label{eq:dualMap}
   		\Phi^*(X)
   		=
   		\operatorname{Tr}_2\!\bigl[
   		(\mathbf 1\otimes \sigma)\,
   		U^\dagger (X\otimes \mathbf 1)U
   		\bigr].
   	\end{equation}
   	For $X = f_J^\dagger f_I$ we have
   	\begin{equation*}
   		f_J^\dagger f_I \otimes \mathbf 1
   		=
   		a_J^\dagger a_I,
   	\end{equation*}
   	since $a_j = f_j \otimes \mathbf 1$. Substituting this into Eq.~\eqref{eq:dualMap}, we have
   	\begin{equation}\label{eq:dualActionOnMoments}
   		\Phi^*(f_J^\dagger f_I)
   		=
   		\operatorname{Tr}_2\!\bigl[
   		(\mathbf 1\otimes \sigma)\,
   		(U^\dagger a_J U)^\dagger\,
   		(U^\dagger a_I U)
   		\bigr].
   	\end{equation}
   	
   	Using the assumption \eqref{eq:linearTransform}, we obtain
   	\begin{equation*}
   		U^\dagger a_I U = (A a + B b)_I,
   		\qquad
   		U^\dagger a_J U = (A a + B b)_J.
   	\end{equation*}
   	Applying Lemma~\ref{lem:femionicMultinonmial} to both multi-indices $I$ and $J$, we get
   	\begin{align*}
   		U^\dagger a_I U
   		&=
   		\sum_{|I|=|L|+|\Omega|}
   		\det\!\bigl(A_{I\times L}\,|\,B_{I\times \Omega}\bigr)\,
   		a_L b_\Omega,
   		\\
   		U^\dagger a_J U
   		&=
   		\sum_{|J|=|K|+|\Xi|}
   		\det\!\bigl(A_{J\times K}\,|\,B_{J\times \Xi}\bigr)\,
   		a_K b_\Xi.
   	\end{align*}
   	Therefore,
   	\begin{equation*}
   		(U^\dagger a_J U)^\dagger
   		=
   		\sum_{|J|=|K|+|\Xi|}
   		\det\!\bigl(A_{J\times K}\,|\,B_{J\times \Xi}\bigr)^*\,
   		b_\Xi^\dagger a_K^\dagger.
   	\end{equation*}
   	
   	Substituting these expressions into Eq.~\ref{eq:dualActionOnMoments} for $\Phi^*(f_J^\dagger f_I)$ yields
   	\begin{equation*}
   		\Phi^*(f_J^\dagger f_I)
   		=
   		\sum_{\substack{|J|=|K|+|\Xi|\\|I|=|L|+|\Omega|}}
   		\det\!\bigl(A_{J\times K}\,|\,B_{J\times \Xi}\bigr)^*\,
   		\det\!\bigl(A_{I\times L}\,|\,B_{I\times \Omega}\bigr)
   		%		\nonumber\\
   		%		&\quad\times
   		\operatorname{Tr}_2\!\bigl[
   		(\mathbf 1\otimes \sigma)\,
   		b_\Xi^\dagger a_K^\dagger a_L b_\Omega
   		\bigr].
   	\end{equation*}
   	
   	Since the families $a$ and $b$ mutually anticommute, each operator
   	$b_\alpha^\dagger$ anticommutes with each operator $a_k^\dagger$ and $a_l$.
   	Hence
   	\begin{equation*}
   		b_\Xi^\dagger a_K^\dagger
   		=
   		(-1)^{|\Xi|\,|K|}\,
   		a_K^\dagger b_\Xi^\dagger,
   		\qquad
   		b_\Xi^\dagger a_L
   		=
   		(-1)^{|\Xi|\,|L|}\,
   		a_L b_\Xi^\dagger.
   	\end{equation*}
   	Combining these relations, we obtain
   	\begin{equation*}
   		b_\Xi^\dagger a_K^\dagger a_L b_\Omega
   		=
   		(-1)^{|\Xi|(|K|+|L|)}\,
   		a_K^\dagger a_L\, b_\Xi^\dagger b_\Omega.
   	\end{equation*}
   	Thus,
   	\begin{equation}\label{eq:dualActionExpand}
   		\Phi^*(f_J^\dagger f_I)
   		=
   		\sum_{\substack{|J|=|K|+|\Xi|\\|I|=|L|+|\Omega|}}
   		(-1)^{|\Xi|(|K|+|L|)}
   		\det\!\bigl(A_{J\times K}\,|\,B_{J\times \Xi}\bigr)^*\,
   		\det\!\bigl(A_{I\times L}\,|\,B_{I\times \Omega}\bigr)
   		%		\nonumber\\
   		%		&\quad\times
   		\operatorname{Tr}_2\!\bigl[
   		(\mathbf 1\otimes \sigma)\,
   		a_K^\dagger a_L\, b_\Xi^\dagger b_\Omega
   		\bigr].
   	\end{equation}
   	
   	Now, from Eq.~\eqref{eq:parityTrick} we have $	b_\Xi
   	=
   	P^{|\Xi|}\otimes f_\Xi$, $	b_\Xi^\dagger
   	=
   	P^{|\Xi|}\otimes f_\Xi^\dagger$
   	and similarly $	b_\Omega
   	=
   	P^{|\Omega|}\otimes f_\Omega$. Therefore,
   	\begin{equation*}
   		b_\Xi^\dagger b_\Omega
   		=
   		P^{|\Xi|+|\Omega|}
   		\otimes
   		f_\Xi^\dagger f_\Omega.
   	\end{equation*}
   	and similarly
   	\begin{equation*}
   		a_K^\dagger a_L
   		=
   		f_K^\dagger f_L \otimes \mathbf 1.
   	\end{equation*}
   	Combining these expressions, we obtain 
   	\begin{equation*}
   		a_K^\dagger a_L\, b_\Xi^\dagger b_\Omega
   		=
   		f_K^\dagger f_L\, P^{|\Xi|+|\Omega|}
   		\otimes
   		f_\Xi^\dagger f_\Omega.
   	\end{equation*}
   	Taking the partial trace over the second tensor factor gives
   	\begin{equation*}
   		\operatorname{Tr}_2\!\bigl[
   		(\mathbf 1\otimes \sigma)\,
   		a_K^\dagger a_L\, b_\Xi^\dagger b_\Omega
   		\bigr]
   		=
   		\operatorname{Tr}\!\bigl(\sigma f_\Xi^\dagger f_\Omega\bigr)\,
   		f_K^\dagger f_L\, P^{|\Xi|+|\Omega|}.
   	\end{equation*}
   	Substituting it into Eq.~\eqref{eq:dualActionExpand} and taking into account Eq.~\eqref{eq:envMoments} we arrive at Eq.~\eqref{eq:Phi-general-compact-parity}.
   \end{proof}
   
   Let us emphasize that in Eq.~\eqref{eq:defOfCPmap} we use just usual tensor product in the  specific representation defined by Eq.~\eqref{eq:parityTrick} rather than fermionic tensor product \cite{Szalay2021}, which is representation invariant. Note that Eq.~\eqref{eq:defOfCPmap} defines a completely positive map, because tensor multiplication by a non-negative matrix, unitary transformation and partial trace are completely positive. Actually, Eq.~\eqref{eq:defOfCPmap} is just a variant of Stinespring dilation of the map $\Phi$. In the case, $\operatorname{Tr}\sigma = 1$ Eq.~\eqref{eq:defOfCPmap} defines a completely positive and trace preserving map, or quantum channel, and in general it is a quantum channel up to a constant $\operatorname{Tr}\sigma $. Explicit formula \eqref{eq:Phi-general-compact-parity} allows us to calculate averages of arbitrary normally ordered monomials  $f_J^\dagger f_I$ at the output of such a channel based on averages of these and lower order monomials $	f_K^\dagger f_L $ and lower order terms of the form $	f_K^\dagger f_L P$ at the input of the channel. But the parity operator $P$ contains contributions from  normally ordered monomials of arbitrary length
   \begin{equation*}
   	P
   	=
   	\sum_{J}  (-2)^{|J|}
   	f_J^\dagger f_J,
   \end{equation*}
   where the sum runs over all subsets   $
   J \subseteq \{1,\dots,m\}
   $ (including empty set). It prevents one from expressing output moments solely in terms of input moments of the same or lower order. This problem is very close to the fact that pure fermionic states should satisfy parity superselection rules to have meaningful reduced density matrices \cite{Amosov2017} independent of specific representation \eqref{eq:parityTrick}. Nevertheless, the fermionic states beyond parity superselection rules can arise in the case, when fermions are not fundamental, but just represent a multi-qubit systems \cite{Lyu2025}.
   
   Nevertheless, if one  assumes that $\sigma$ is even, i.e. $	[P,\sigma]=0$, then 
   \begin{equation*}
   	\Gamma_{\Xi;\Omega}
   	=
   	\operatorname{Tr}\!\bigl(\sigma f_\Xi^\dagger f_\Omega\bigr)
   	=0
   	\qquad
   	\text{whenever } |\Xi|+|\Omega| \text{ is odd}.
   \end{equation*}
   So the terms of the form $f_K^\dagger f_L\, P$ vanish from Eq.~\eqref{eq:Phi-general-compact-parity} and we obtain the following theorem.
   
   \begin{theorem}\label{th:Phi-general-compact-even-sigma}
   	In the setting of Lemma~\ref{lem:Phi-general-compact-parity}, assume in addition that $[P,\sigma]=0$, then
   	\begin{equation}
   		\Phi^*\!\bigl(f_J^\dagger f_I\bigr)
   		=
   		\sum_{\substack{
   				|J|=|K|+|\Xi|\\
   				|I|=|L|+|\Omega|\\
   				|\Xi|+|\Omega|  =0 \, ( \!\!\!\!\!\! \mod 2)
   		}}
   		(-1)^{|\Xi|(|K|+|L|)}
   		\det\!\bigl( A_{J\times K}\,|\,B_{J\times \Xi} \bigr)^*\,
   		\Gamma_{\Xi;\Omega}\,
   		\det\!\bigl( A_{I\times L}\,|\,B_{I\times \Omega} \bigr)
   		f_K^\dagger f_L .
   		\label{eq:Phi-general-compact-even-sigma}
   	\end{equation}
   	where the sum runs over all ordered multi-indices $
   	K,L,\Xi,\Omega \subseteq \{1,\dots,m\}$ satisfying the constraints specified under the summation sign.
   \end{theorem}
   
   Although it is an immediate corollary of Lemma~\ref{lem:Phi-general-compact-parity}, we call it a theorem, because it has much greater practical importance than the general result given by Eq.~\eqref{eq:Phi-general-compact-parity}. Now for any integer $k$ the  space $	\operatorname{span} \{ f_J^\dagger f_I : |I| +|J| \leq k\}$  is invariant under $\Phi^*$. So any output moment of creation and annihilation operators is fully defined by same or lower order input moments only. This makes computation of the low order moments computationally efficient.

   Now let us consider several special cases of Eq.~\eqref{eq:Phi-general-compact-even-sigma}.
   First, assume that the environment is in the vacuum state $\sigma = |0\rangle\langle 0|$. Then
   \begin{equation*}
   	\Gamma_{\Xi;\Omega}
   	=
   	\langle 0| f_\Xi^\dagger f_\Omega |0\rangle
   	=
   	\delta_{\Xi,\varnothing}\,\delta_{\Omega,\varnothing},
   \end{equation*}
   and therefore only the term $\Xi=\Omega=\varnothing$ contributes in \eqref{eq:Phi-general-compact-even-sigma}. As a result,
   \begin{equation*}
   	\Phi^*\!\bigl(f_J^\dagger f_I\bigr)
   	=
   	\sum_{\substack{|K|=|J|\\ |L|=|I|}}
   	\det\!\bigl(A_{J\times K}\bigr)^*\,
   	\det\!\bigl(A_{I\times L}\bigr)\,
   	f_K^\dagger f_L .
   \end{equation*}
   So now the spaces spanned by  $f_J^\dagger f_I$ with fixed $|J|$,  $|I|$, and hence the output moments of fixed order are defined  by the moments of the same order only similarly to \cite{Teretenkov20,NosTer20,Barthel2022,Ivanov2022}.

   Second, assume that the environment state $\sigma$ is a gauge-invariant Gaussian state, whose nonzero second moments are encoded by the matrix
   \begin{equation*}
   	C_{\alpha\beta}
   	=
   	\operatorname{Tr}\!\bigl(\sigma\, f_\beta^\dagger f_\alpha\bigr).
   \end{equation*}
   Then all higher moments are given by the Wick theorem  as (see, e.g. Eq.~\cite[Eq. (17)]{Bravyi2005} for an equivalent formula in terms of  Pfaffian of an enlarged antisymmetric matrix)
   \begin{equation*}
   	\Gamma_{\Xi;\Omega}
   	=
   	\delta_{|\Xi|,|\Omega|}\,
   	\det\!\bigl(C_{\Omega\times \Xi}\bigr),
   \end{equation*}
   where $C_{\Omega\times \Xi}$ denotes the  submatrix of $C$ formed by taking rows indexed by $\Omega$ and columns indexed by $\Xi$. Substituting this into Eq.~\eqref{eq:Phi-general-compact-even-sigma}, we obtain
   \begin{equation*}
   	\Phi^*\!\bigl(f_J^\dagger f_I\bigr)
   	=
   	\sum_{\substack{
   			|J|=|K|+r\\
   			|I|=|L|+r
   	}}
   	(-1)^{r(|K|+|L|)}
   	\sum_{\substack{|\Xi|=|\Omega|=r}}
   	\det\!\bigl( A_{J\times K}\,|\,B_{J\times \Xi} \bigr)^*\,
   	\det\!\bigl(C_{\Omega\times \Xi}\bigr)\,
   	\det\!\bigl( A_{I\times L}\,|\,B_{I\times \Omega} \bigr)
   	f_K^\dagger f_L .
   \end{equation*}
   
   So even for the  gauge-invariant Gaussian channels  the formula for arbitrary moments of creation and annihilation operators is not much simpler than general formula \eqref{eq:Phi-general-compact-parity}. So we can consider our class of models to be a natural generalization of  gauge-invariant Gaussian channels and it is still tractable at the same level if one can compute $	\Gamma_{\Xi;\Omega}$, which can be done for the states other than the Gaussian ones. 
   
   In particular, let us start with the fermionic Fock state, which is a pure gauge-invariant Gaussian state,
   \begin{equation*}
   	\sigma = |M\rangle\langle M|,
   	\qquad
   	|M\rangle \equiv f_M^\dagger |0\rangle ,
   \end{equation*}
   where $M$ is an ordered subset of environment modes. Then
   \begin{equation*}
   	\Gamma_{\Xi;\Omega}
   	=
   	\langle M|\, f_\Xi^\dagger f_\Omega \,|M\rangle
   	=
   	\delta_{\Xi,\Omega}\,\mathbf{1}_{\Xi\subseteq M}.
   \end{equation*}
   Hence, Eq.~\eqref{eq:Phi-general-compact-even-sigma} reduces to
   \begin{equation*}
   	\Phi^*\!\bigl(f_J^\dagger f_I\bigr)
   	=
   	\sum_{\Xi\subseteq M}
   	\sum_{\substack{|J|= |K| +|\Xi|\\ |I|= |L| + |\Xi|}}
   	(-1)^{|\Xi|(|K|+|L|)}
   	\det\!\bigl( A_{J\times K}\,|\,B_{J\times \Xi} \bigr)^*\,
   	\det\!\bigl( A_{I\times L}\,|\,B_{I\times \Xi} \bigr)\,
   	f_K^\dagger f_L .
   \end{equation*}
   
   Similarly, all the moments of creation and annihilation operators can be calculated for non-Gaussian states that are diagonal in the Fock basis  $|M\rangle\langle M|$, e.g. uniform averages of the states with a fixed number of particles:
   \begin{equation*}
   	\sigma
   	=
   	\frac{1}{\binom{m}{N}}
   	\sum_{|M|=N} |M\rangle\langle M|.
   \end{equation*}
   Then
   \begin{equation*}
   	\Gamma_{\Xi;\Omega}
   	=
   	\delta_{\Xi,\Omega}
   	\frac{\binom{m-|\Xi|}{\,N-|\Xi|\,}}{\binom{m}{N}}
   \end{equation*}
   and Eq.~\eqref{eq:Phi-general-compact-even-sigma} takes the form
   \begin{equation*}
   	\Phi^*\!\bigl(f_J^\dagger f_I\bigr)
   	=
   	\sum_{\substack{
   			|J|=|K|+|\Xi|\\
   			|I|=|L|+|\Xi|
   	}}
   	(-1)^{|\Xi|(|K|+|L|)}
   	\frac{\binom{m-|\Xi|}{\,N-|\Xi|\,}}{\binom{m}{N}}
   	\det\!\bigl( A_{J\times K}\,|\,B_{J\times \Xi} \bigr)^*\,
   	\det\!\bigl( A_{I\times L}\,|\,B_{I\times \Xi} \bigr)
   	f_K^\dagger f_L.
   \end{equation*}
   
   Thus, Theorem~\ref{th:Phi-general-compact-even-sigma} provides a convenient way to construct wide classes of fermionic channels with good structural properties. Namely, their action on normally ordered monomials is explicit, and moments up to any fixed order evolve in a closed way, independently of higher-order moments. 
   
   \section{Second quantization candidates for  non-Hermitian systems}
   \label{sec:secondQuant}
   
   In this section we examine whether fermionic completely positive maps with linear action on annihilation operators can be interpreted as a form of second quantization of non-Hermitian one-particle dynamics. Here we assume that $\sigma$ is normalized, so $\Gamma_{\varnothing;\varnothing}=\operatorname{Tr}\sigma=1$. For $J=\varnothing$ and $I=\{j\}$, the general formula~\eqref{eq:Phi-general-compact-even-sigma} reduces to
   \begin{equation*}
   	\Phi^*(f_j)=\sum_l A_{jl} f_l .
   \end{equation*}
   Equivalently, in vector form,
   \begin{equation*}
   	\Phi^*(f)=Af .
   \end{equation*}

   A matrix $A$ is admissible (i.e., can appear in $\Phi^*(f)=Af$ for $\Phi$ defined by Eq.~\eqref{eq:defOfCPmap})
   if and only if it is a contraction:
   \begin{equation*}
   	A A^\dagger \le \mathbf 1,
   \end{equation*}
   because in this case one can choose
   \begin{equation}\label{eq:BDef}
   	B = \bigl(\mathbf 1 - A A^\dagger\bigr)^{1/2},
   \end{equation}
   to satisfy Eq.~\eqref{eq:isometry-condition-AB}.
   
   Thus, the channel $\Phi$ can be interpreted as second quantization of one-particle contraction $A$ (on the level of one-particle operators). This shows that, for open quantum systems, second quantization is highly non-unique. In particular, one can consider a one-parametric contractive semigroup, 
   \begin{equation}\label{eq:contractiveSemigroup}
   	A(t) = e^{i H_{\rm eff} t},
   \end{equation}
   which leads to Schroedinger equation for $f_t =  e^{i H_{\rm eff} t} f$ with non-Hermitian Hamiltonian $i(H_{\rm eff} -H_{\rm eff}^{\dagger})\leqslant 0$ (similarly to \cite{Teretenkov2020}). However, general moments are not described by a semigroup. In fact, without additional assumptions, even the multiparticle dynamics is not a semigroup:
   \begin{equation}\label{eq:badSecQuant}
   	\Phi^*(f_I)
   	=
   	\sum_{\substack{
   			|I|=|L|+|\Omega|\\
   			|\Omega|\equiv 0 \,(\mathrm{mod}\,2)
   	}}
   	\Gamma_{\varnothing;\Omega}\,
   	\det\!\bigl((e^{i H_{\rm eff} t})_{I\times L}\,|\,  \bigl(\mathbf 1 - e^{i H_{\rm eff} t} e^{-i H_{\rm eff}^{\dagger} t}\bigr)^{1/2} \left. \right._{I\times \Omega}\bigr)\,
   	f_L ,
   \end{equation}
   so one can say that it is not a genuine candidate for second quantization of non-interacting fermions with non-Hermitian one-particle Hamiltonian (with dissipative $i H_{\rm eff}$).
   
   But assume in addition that the environment state is gauge-invariant $[\sigma, N_f] = 0$. Then $	\Gamma_{\varnothing;\Omega}=0$ for all $
   \Omega\neq\varnothing$, and therefore
   \begin{equation}
   	\Phi^*(f_I)
   	=
   	\sum_{|L|=|I|}
   	\det\!\bigl(A_{I\times L}\bigr)\,f_L  \equiv \sum_{|L|=|I|}
   	(\wedge^{|I|} A)_{I,L} \,f_L ,
   	\label{eq:Phi-fI-gauge}
   \end{equation}
   where $\wedge^{p} A$ is an exterior power of $A$, or, equivalently,
   \begin{equation*}
   	\Phi^*(f^{\otimes p}) =  A^{\otimes p} f^{\otimes p},
   \end{equation*}
   where $(f^{\otimes p})_{i_1,\dots,i_p}\equiv f_{i_1}\cdots f_{i_p}$ is a tensor of all possible products of annihilation operators, not only ones with ordered indices. 
   If now we take  a semigroup defined by Eq.~\eqref{eq:contractiveSemigroup} for $A$, then for every $p=0,1,\dots,m$ we have
   \begin{equation*}
   	\Phi_t^*(f^{\otimes p}) = 	e^{\,i t\,  H_{\rm eff}^{(p)}} f^{\otimes p},
   \end{equation*}
   where
   \begin{equation}\label{eq:secQunat}
   	H_{\rm eff}^{(p)}
   	\equiv
   	\sum_{r=1}^p
   	\mathbf 1^{\otimes(r-1)} \otimes H_{\rm eff} \otimes \mathbf 1^{\otimes(p-r)},
   \end{equation}
   which coincides with formulae for second-quantized non-interacting fermions in the case of self-adjoint Hamiltonian $H_{\rm eff}$. Thus, in contrast to Eq.~\eqref{eq:badSecQuant}, this does define a genuine second quantization of non-interacting fermions with non-Hermitian one-particle Hamiltonian (with dissipative $i H_{\rm eff}$).
   
   Let us emphasize that for general ordered monomials $f_J^\dagger f_I$ formulae \eqref{eq:Phi-general-compact-even-sigma}, \eqref{eq:BDef}, \eqref{eq:contractiveSemigroup} may not define a semigroup similarly to Eq.~\eqref{eq:badSecQuant} even in the gauge-invariant case. However, we do not view this as a drawback of such quantizations, but as an important property of open quantum systems: the quantization can break semigroup properties for some observables which describe correlations. 
   
   \section{Dynamics averaged over a Poisson process}
   \label{sec:dynamics}
   
   In this section, we show that the constructed family of completely positive maps can be used to define a completely positive one-parameter trace-preserving semigroup such that the equations for the moments of creation  and annihilation operators  up to any fixed order are closed. Namely, starting from any completely positive map, one can \cite[Lemma 1]{Wolf2008} construct a GKSL generator by the formula
   \begin{equation}\label{eq:lindGenByPhi}
   	\mathcal L(\rho)
   	\equiv
   	\Phi(\rho)-\frac12\{\Phi^*(\mathbf 1),\rho\},
   \end{equation}
   because the Kraus representation for $\Phi$ defines a Lindblad representation for $\mathcal L$. But let us emphasize we need  not to know an explicit Lindblad representation to ensure that Eq.~\eqref{eq:lindGenByPhi} defines a GKSL generator. For the completely positive map defined by Eq.~\eqref{eq:Phi-general-compact-even-sigma} we have
   \begin{equation*}
   	\Phi^*(\mathbf 1)=\Gamma_{\varnothing;\varnothing}\,\mathbf 1,
   \end{equation*}
   so the GKSL generator \eqref{eq:lindGenByPhi} in the Heisenberg representation takes the form
   \begin{equation*}
   	\mathcal L^*(X)
   	=
   	\Phi^*(X)- \Gamma_{\varnothing;\varnothing} X.
   \end{equation*}
   Note that, such a GKSL generator naturally arises if one considers a discrete-time semigroup whose Heisenberg dynamics during each time step is described by a quantum channel $\Gamma_{\varnothing;\varnothing}^{-1}\Phi^*$ and then averages it with respect to a Poisson process with intensity $ \Gamma_{\varnothing;\varnothing}$.
   
   If a linear space of operators is invariant under $\Phi^*$, then it automatically does for $\mathcal L^*$. More precisely we obtain the following lemma.
   
   \begin{lemma}
   	Let $\Phi: \mathbb{C}^{2^m \times 2^m} \rightarrow \mathbb{C}^{2^m \times 2^m}$ be the completely positive map whose  action in Heisenberg picture on $f_J^\dagger f_I$ is given by Eq.~\eqref{eq:Phi-general-compact-even-sigma}. Define a GKSL generator by Eq.~\eqref{eq:lindGenByPhi}. Then
   	\begin{align*}
   		\mathcal L^*\!\bigl(f_J^\dagger f_I\bigr)
   		&=
   		\sum_{\substack{
   				|J|=|K|+|\Xi|\\
   				|I|=|L|+|\Omega|\\
   				|\Xi|+|\Omega| \equiv 0 \,(\mathrm{mod}\,2)
   		}}
   		(-1)^{|\Xi|(|K|+|L|)}
   		\det\!\bigl( A_{J\times K}\,|\,B_{J\times \Xi} \bigr)^*\,
   		\Gamma_{\Xi;\Omega}\,
   		\det\!\bigl( A_{I\times L}\,|\,B_{I\times \Omega} \bigr)
   		\,f_K^\dagger f_L
   		\\
   		&\qquad
   		-\Gamma_{\varnothing;\varnothing}\,f_J^\dagger f_I .
   	\end{align*}
   \end{lemma}
   
   As an immediate corollary we obtain the following theorem. 
   
   \begin{theorem}\label{th:GKSLmoments}
   	Let $(\rho_t)_{t\ge 0}$ be a solution of the master equation
   	\begin{equation}\label{eq:GKSLeq}
   		\frac{d}{dt}\rho_t=\mathcal L(\rho_t),
   	\end{equation}
   	where the generator $\mathcal L$ is defined by Eq.~\eqref{eq:lindGenByPhi}. For arbitrary subsets $I,J\subset\{1,\dots,m\}$, define the normally ordered moments
   	\begin{equation*}
   		\langle  f_J^\dagger f_I \rangle_t \equiv \operatorname{Tr}\!\bigl(\rho_t\, f_J^\dagger f_I\bigr).
   	\end{equation*}
   	Then the family $ \langle  f_J^\dagger f_I \rangle_t $ satisfies the  linear system of differential equations
   	\begin{align}
   		\frac{d}{dt} \langle  f_J^\dagger f_I \rangle_t
   		&=
   		\sum_{\substack{
   				|J|=|K|+|\Xi|\\
   				|I|=|L|+|\Omega|\\
   				|\Xi|+|\Omega| \equiv 0 \,(\mathrm{mod}\,2)
   		}}
   		(-1)^{|\Xi|(|K|+|L|)}
   		\det\!\bigl( A_{J\times K}\,|\,B_{J\times \Xi} \bigr)^*\,
   		\Gamma_{\Xi;\Omega}\,
   		\det\!\bigl( A_{I\times L}\,|\,B_{I\times \Omega} \bigr)
   		\,  \langle  f_K^\dagger f_L \rangle_t
   		\nonumber\\
   		&\qquad
   		-\Gamma_{\varnothing;\varnothing}\, \langle  f_J^\dagger f_I \rangle_t \label{eq:closedEqForMoments}
   	\end{align}
   	and, for each fixed order, the corresponding subsystem of equations is closed.
   \end{theorem}
   Thus, to compute moments up to a fixed order, one needs to evaluate an exponential of a matrix, whose dimension grows polynomially in $m$, instead of solving Eq.~\eqref{eq:GKSLeq}, which is also finite dimensional linear equation and its solution reduces to an exponential of a matrix, but with  size which is exponential in  $m$. Analogously to the analysis in \cite{NosTer20} one can also calculate Markovian (defined by regression formulae) multi-time correlation functions of fixed order $p$ for creation and annihilation operators in terms of $p$ exponentials of matrices with polynomial sizes in  $m$.  
   
   For a gauge-invariant matrix $\sigma$ similarly to Eq.~\eqref{eq:Phi-fI-gauge} we have
   \begin{equation*}
   	\mathcal L^*(f_I)
   	=
   	\Gamma_{\varnothing;\varnothing} \left(
   	\sum_{|L|=|I|}
   	\det\!\bigl(A_{I\times L}\bigr)\,f_L
   	-
   	\,f_I.\right)
   \end{equation*}
   and, hence, Eq.~\eqref{eq:closedEqForMoments} for $f_I$ takes the form
   \begin{equation}\label{eq:dynCreat}
   	\langle  f^{\otimes p} \rangle_t =
   	e^{\,t\,\Gamma_{\varnothing;\varnothing}\,( A^{\otimes p}-\mathbf 1)}
   	\langle  f^{\otimes p} \rangle_0.
   \end{equation}
   But for general  gauge-invariant state $\sigma$ the moments involving both creation and annihilation operators satisfy the general form with contributions involving the matrix $B$:
   \begin{equation*}
   	\frac{d}{dt} \langle  f_J^\dagger f_I \rangle_t
   	=
   	\sum_{\substack{
   			|J|=|K|+|\Xi|\\
   			|I|=|L|+|\Omega|\\
   			|\Xi|=|\Omega|
   	}}
   	(-1)^{|\Xi|(|K|+|L|)}
   	\det\!\bigl( A_{J\times K}\,|\,B_{J\times \Xi} \bigr)^*\,
   	\Gamma_{\Xi;\Omega}\,
   	\det\!\bigl( A_{I\times L}\,|\,B_{I\times \Omega} \bigr)
   	\, \langle f_K^\dagger f_L \rangle_t
   	- \Gamma_{\varnothing;\varnothing}\, \langle  f_J^\dagger f_I \rangle_t .
   \end{equation*}
   Nevertheless for $\sigma = 	\Gamma_{\varnothing;\varnothing} |0\rangle \langle 0|$ it reduces to
   \begin{equation*}
   	\frac{d}{dt}\langle f_J^\dagger f_I\rangle_t
   	=
   	\Gamma_{\varnothing;\varnothing}\left(	\sum_{\substack{|K|=|J|\\|L|=|I|}}
   	(\wedge^{|J|} A)_{J,K}^*\,
   	(\wedge^{|I|} A)_{I,L}\,
   	\langle f_K^\dagger f_L\rangle_t
   	-
   	\langle f_J^\dagger f_I\rangle_t\right),
   \end{equation*}
   so we have
   \begin{equation*}
   	\,\bigl\langle (f^\dagger)^{\otimes p}\otimes f^{\otimes q}\bigr\rangle_t
   	=
   	\exp \left(	\Gamma_{\varnothing;\varnothing} \Bigl[(A^{\dagger})^{\otimes p}\otimes A^{\otimes q}-\mathbf 1\Bigr]t\right)
   	\bigl\langle (f^\dagger)^{\otimes p}\otimes f^{\otimes q}\bigr\rangle_0.
   \end{equation*}
   
   Remark that we can consider sums of GKSL generators with closed dynamics up to fixed order to obtain new generators with such a property. For example, one can sum several GKSL generators from Theorem~\ref{th:GKSLmoments}. One can also simply add to such a GKSL generator a generator of unitary second-quantized dynamics  with one-particle Hamiltonian $H \in \mathbb{C}^{m \times m}$. In the latter case for a gauge-invariant matrix $\sigma$ Eq.~\eqref{eq:dynCreat} is replaced by
   \begin{equation*}
   	\langle  f^{\otimes p} \rangle_t =
   	e^{\,t\,\left(iH^{(p)} +\Gamma_{\varnothing;\varnothing}\,[A^{\otimes p}-\mathbf 1] \right)}
   	\langle  f^{\otimes p} \rangle_0,
   \end{equation*}
   where $H^{(p)}$ is defined similarly to Eq.~\eqref{eq:secQunat}. Note that as $A$ is a contraction, then the term $\Gamma_{\varnothing;\varnothing}\,[A^{\otimes p}-\mathbf 1]$ is dissipative. But while  $H^{(p)}$ is a sum of local terms, this dissipative contribution is $p$-particle interaction. Such a $p$-particle interaction is not usual for closed quantum many-body systems, but can arise naturally from averaging of local fluctuating interactions, if they are described by the Poisson process.
   
   \section{Post-selected linear fermionic maps}
   \label{sec:posSec}
   
   In different applications not only the completely positive and trace preserving maps are important, but general completely positive maps as well. For example such maps can arise from post-selection of the unitary dynamics of an open system  on some  environment measurement result described by an effect $E$, i.e. selfadjoint matrix such that $0 \leqslant E  \leqslant \mathbf 1$. The corresponding density-matrix dynamics is described by the map
   \begin{equation*}
   	\Phi_E(\rho)
   	\equiv
   	\operatorname{Tr}_2\!\bigl[
   	(\mathbf 1\otimes E)\,
   	U(\rho\otimes \sigma)U^\dagger
   	\bigr].
   \end{equation*}
   whose complete positivity follows directly from the representation
   \begin{equation*}
   	\Phi_E(\rho)
   	=
   	\operatorname{Tr}_2\!\bigl[
   	(\mathbf 1\otimes E^{1/2})\,U(\rho\otimes \sigma)U^\dagger(\mathbf 1\otimes E^{1/2})
   	\bigr].
   \end{equation*}
   $	\Phi_E$ is a direct generalization of $	\Phi$ defined by Eq.~\eqref{eq:defOfCPmap}.
   
   \begin{lemma}
   	\label{lem:PhiE-general}
   	Assume that a unitary $U: \mathbb{C}^{2^m} \otimes \mathbb{C}^{2^m} \rightarrow \mathbb{C}^{2^m} \otimes \mathbb{C}^{2^m}$  preserves the total particle number and acts linearly as
   	\begin{equation}\label{eq:unitaryRep}
   		U^\dagger
   		\begin{pmatrix}
   			a\\ b
   		\end{pmatrix}
   		U
   		=
   		W
   		\begin{pmatrix}
   			a\\ b
   		\end{pmatrix},
   		\qquad
   		W=
   		\begin{pmatrix}
   			A & B\\
   			C & D
   		\end{pmatrix},
   	\end{equation}
   	where $W$ is a unitary $2m\times 2m$ matrix. Let $\sigma: \mathbb{C}^{2^m}  \rightarrow \mathbb{C}^{2^m}$, $\sigma = \sigma^{\dagger}$, $\sigma  \ge 0$ act on the second tensor factor. Let $E$ be an even selfadjoint non-negative operator with expansion
   	\begin{equation*}
   		E
   		=
   		\sum_{ |M|+|N|\equiv 0\!\!\!\pmod 2}
   		e_{M;N}\, f_M^\dagger f_N ,
   		\qquad
   		e_{M;N}^*=e_{N;M}.
   	\end{equation*}
   	
   	Define a completely positive map
   	\begin{equation}\label{eq:defPhiE}
   		\Phi_E(\rho)
   		\equiv
   		\operatorname{Tr}_2\!\bigl[
   		(\mathbf 1\otimes E)\,
   		U(\rho\otimes \sigma)U^\dagger
   		\bigr].
   	\end{equation}
   	
   	Then for any ordered multi-indices $I,J$,
   	\begin{equation}\label{eq:actionOfPhiE}
   		\Phi_E^*\!\bigl(f_J^\dagger f_I\bigr)
   		=
   		\sum_{\substack{
   				|M|+|N|\equiv 0\!\!\!\pmod 2\\
   				|J|+|M|=|K|+|\Xi|\\
   				|I|+|N|=|L|+|\Omega|
   		}}
   		(-1)^{(|I|+|J|)|M|+|\Xi|(|K|+|L|)}\,
   		e_{M;N}\,
   		\bigl(\Delta^{W}_{J,M;\,K,\Xi}\bigr)^*\,
   		\Gamma_{\Xi;\Omega}\,
   		\Delta^{W}_{I,N;\,L,\Omega}
   		\,
   		f_K^\dagger f_L\, P^{|\Xi|+|\Omega|}.
   	\end{equation}
   	where the sum runs over all ordered multi-indices $M,N,K,L,\Xi,\Omega \subseteq \{1,\dots,m\} $ satisfying the constraints specified under the summation sign, $	\Gamma_{\Xi;\Omega}
   	\equiv
   	\operatorname{Tr}\!\bigl(
   	\sigma\, f_\Xi^\dagger f_\Omega
   	\bigr)$ and
   	\begin{equation}
   		\Delta^{W}_{I,N;\,L,\Omega}
   		\equiv
   		\det
   		\begin{pmatrix}
   			A_{I\times L} & B_{I\times \Omega}\\[2mm]
   			C_{N\times L} & D_{N\times \Omega}
   		\end{pmatrix},
   	\end{equation}
   	for multi-indices $I,\Omega,N,L$ satisfying $|I|+|N|=|L|+|\Omega|$.
   \end{lemma}
   
   \begin{proof}
   	
   	The proof  is similar to the one of Lemma \ref{lem:Phi-general-compact-parity}, so we focus only on the main steps. The dual map to $\Phi_E$ defined by Eq.~\eqref{eq:defPhiE} is
   	\begin{equation*}
   		\Phi_E^*(X)
   		=
   		\operatorname{Tr}_2\!\bigl[(\mathbf 1\otimes \sigma)\,U^\dagger (X\otimes E) U\bigr]
   	\end{equation*}
   	and we apply it to $	X=f_J^\dagger f_I$, taking into account $	f_J^\dagger f_I\otimes \mathbf 1 = a_J^\dagger a_I$. Because $E$ is even, $	\mathbf 1\otimes E$ can be rewritten in terms of normally ordered monomials in $b$ without parity operators as
   	\begin{equation*}
   		\mathbf 1\otimes E
   		=
   		\sum_{|M|+|N|\equiv 0\!\!\!\pmod 2}
   		e_{M;N}\, b_M^\dagger b_N .
   	\end{equation*}
   	Then
   	\begin{equation*}
   		(f_J^\dagger f_I)\otimes E
   		=
   		(a_J^\dagger a_I)(\mathbf 1\otimes E)
   		=
   		\sum_{ |M|+|N|\equiv 0\!\!\!\!\pmod 2}
   		e_{M;N}\, a_J^\dagger a_I\, b_M^\dagger b_N.
   	\end{equation*}
   	and so using CAR, we have $a_I\, b_M^\dagger = (-1)^{|I||M|} b_M^\dagger a_I$, and,  hence, we obtain
   	\begin{equation}
   		\Phi_E^*(f_J^\dagger f_I)
   		=
   		\sum_{\substack{|M|+|N|\equiv 0\!\!\!\pmod 2}}
   		(-1)^{|I||M|}\, e_{M;N}\,
   		\operatorname{Tr}_2\!\Bigl[
   		(\mathbf 1\otimes \sigma)\,
   		U^\dagger (a_J^\dagger b_M^\dagger a_I b_N)U
   		\Bigr].
   		\label{eq:dual-start}
   	\end{equation}
   	
   	Using Eq.~\eqref{eq:unitaryRep} similar to Lemma \ref{lem:femionicMultinonmial} we have
   	\begin{equation*}
   		U^\dagger(a_I b_N)U
   		=
   		\sum_{\substack{|L|+|\Omega|=|I|+|N|}}
   		\Delta^W_{I,N;\,L,\Omega}\,
   		a_L b_\Omega.
   	\end{equation*}
   	Taking the adjoint of the corresponding formula for $a_J b_M$, and using $a_J^\dagger b_M^\dagger
   	=
   	(-1)^{|J||M|}(a_J b_M)^\dagger$,
   	we obtain
   	\begin{equation*}
   		U^\dagger(a_J^\dagger b_M^\dagger)U
   		=
   		\sum_{\substack{|K|+|\Xi|=|J|+|M|}}
   		(-1)^{|J||M|}
   		\bigl(\Delta^W_{J,M;\,K,\Xi}\bigr)^*\,
   		b_\Xi^\dagger a_K^\dagger.
   	\end{equation*}
   	Hence,
   	\begin{align}
   		U^\dagger(a_J^\dagger b_M^\dagger a_I b_N)U
   		&=
   		\sum_{\substack{
   				|J|+|M|=|K|+|\Xi|\\
   				|I|+|N|=|L|+|\Omega|
   		}}
   		(-1)^{|J||M|}
   		\bigl(\Delta^W_{J,M;\,K,\Xi}\bigr)^*\,
   		\Delta^W_{I,N;\,L,\Omega}\,
   		b_\Xi^\dagger a_K^\dagger a_L b_\Omega .
   		\label{eq:middle}
   	\end{align}
   	Substituting Eq.~\eqref{eq:middle} into Eq.~\eqref{eq:dual-start}, we obtain
   	\begin{align}
   		\Phi_E^*(f_J^\dagger f_I)
   		&=
   		\sum_{\substack{|M|+|N|\equiv 0\!\!\!\pmod 2}}
   		e_{M;N}
   		\sum_{\substack{
   				|J|+|M|=|K|+|\Xi|\\
   				|I|+|N|=|L|+|\Omega|
   		}}
   		(-1)^{|I||M|+|J||M|}
   		\bigl(\Delta^W_{J,M;\,K,\Xi}\bigr)^*\,
   		\Delta^W_{I,N;\,L,\Omega}
   		\nonumber\\
   		&\qquad\qquad\times
   		\operatorname{Tr}_2\!\bigl[
   		(\mathbf 1\otimes \sigma)\,
   		b_\Xi^\dagger a_K^\dagger a_L b_\Omega
   		\bigr]. \label{eq:actionOfPhiEInTremsOfab}
   	\end{align}
   	
   	Finally, using $b_\Xi^\dagger a_K^\dagger a_L b_\Omega = (-1)^{|\Xi|(|K|+|L|)} \, a_K^\dagger a_L \, b_\Xi^\dagger b_\Omega
   	=
   	(-1)^{|\Xi|(|K|+|L|)}
   	f_K^\dagger f_L\, P^{|\Xi|+|\Omega|}
   	\otimes
   	f_\Xi^\dagger f_\Omega$,
   	we have
   	\begin{equation*}
   		\operatorname{Tr}_2\!\bigl[
   		(\mathbf 1\otimes \sigma)\,
   		b_\Xi^\dagger a_K^\dagger a_L b_\Omega
   		\bigr]
   		=
   		(-1)^{|\Xi|(|K|+|L|)}
   		f_K^\dagger f_L\, P^{|\Xi|+|\Omega|}
   		\Gamma_{\Xi;\Omega},
   	\end{equation*}
   	where $	\Gamma_{\Xi;\Omega}
   	\equiv
   	\operatorname{Tr}\!\bigl(\sigma f_\Xi^\dagger f_\Omega\bigr)$.	Thus, Eq.~\eqref{eq:actionOfPhiEInTremsOfab} takes the form of Eq.~\eqref{eq:actionOfPhiE}.
   \end{proof}
   
   For even states analogously  to Theorem \ref{th:Phi-general-compact-even-sigma} we obtain the following theorem.
   
   \begin{theorem}
   	In the setting of Lemma~\ref{lem:PhiE-general}, assume in addition that $[P,\sigma]=0$, then
   	
   	\begin{equation*}
   		\Phi_E^*\!\bigl(f_J^\dagger f_I\bigr)
   		=
   		\sum_{\substack{
   				|M|+|N|\equiv 0\!\!\!\pmod 2\\
   				|J|+|M|=|K|+|\Xi|\\
   				|I|+|N|=|L|+|\Omega|\\
   				|\Xi|+|\Omega|\equiv 0\!\!\!\pmod 2
   		}}
   		(-1)^{(|I|+|J|)|M|+|\Xi|(|K|+|L|)}
   		\,e_{M;N}\,
   		\bigl(\Delta^W_{J,M;\,K,\Xi}\bigr)^*\,
   		\Gamma_{\Xi;\Omega}\,
   		\Delta^W_{I,N;\,L,\Omega}\,
   		f_K^\dagger f_L ,
   	\end{equation*}
   	where the  sum runs over all ordered multi-indices $M,N,K,L,\Xi,\Omega \subseteq \{1,\dots,m\} $.
   \end{theorem}
   
   The above construction naturally extends to the framework of quantum instruments. In particular, if the operator $E$ is interpreted as an effect corresponding to a given measurement outcome and $\operatorname{Tr}\sigma=1$, then the map $\Phi_E$ describes a completely positive trace-non-increasing transformation associated with post-selection on this outcome. More generally, considering a positive operator-valued measure $E(B)$ on the environment, one obtains a family of completely positive maps $\Phi_{E(B)}$ forming a quantum instrument. In this setting, the theorem provides an explicit description of the Heisenberg-picture action of each operation $\Phi_{E(B)}^*$ on fermionic monomials. Therefore, it can be directly applied to analyze conditional dynamics of fermionic systems under measurements on the environment, including selective evolutions corresponding to individual outcomes as well as their combinations obtained by summation over subsets of outcomes.
   
   \section{Conclusions}
   
   We constructed a broad class of completely positive maps for fermionic systems induced by linear transformations of system and environment CAR modes. For these maps, we derived explicit Heisenberg-picture formulas for arbitrary normally ordered monomials in terms of minors of the mode-transformation matrices and correlation tensors of the environment state. This gives a direct algebraic description of the action of the map on many-body observables.
   
   Our main result is that for even environment states the linear span of normally ordered monomials up to any fixed order is invariant under the corresponding dual maps. Therefore, low-order moments evolve independently of higher-order ones, which makes their dynamics efficiently computable. In this way, the developed framework extends the closure property familiar from Gaussian fermionic dynamics to a substantially wider class of non-Gaussian completely positive maps.
   
   We also showed that the same construction produces GKSL generators for which moments up to any fixed order satisfy closed linear systems of equations. This provides an efficient way to study Markovian fermionic dynamics without solving the full master equation on the exponentially large many-body space. Such a phenomenon of having series of invariant spaces of polynomial size in number of modes is called operator space fragmentation and is now actively studied \cite{Essler2020, Li2023, Teretenkov2024b}. In addition, we discussed the relation of the construction to second quantization of non-Hermitian one-particle contractions and showed that in the gauge-invariant case the action on products of annihilation operators is governed by exterior powers of the underlying one-particle map. Finally, we extended the formalism to completely positive maps arising from post-selection on environment measurement outcomes, thus including conditional and instrument-type transformations into the same framework.
   
   Possible directions for future work include extensions to bosonic systems and to general Bogoliubov transformations mixing creation and annihilation operators, in particular, in the context of their recent open quantum systems applications~\cite{Trubilko2026}. It would also be interesting to identify subclasses of the present construction that possess additional composition or semigroup properties, as well as to explore applications to concrete many-body open-system models. We also remark that some GKSL equations can be mapped, via an appropriate vectorization and Wick rotation of the coefficients, to bilayer Schroedinger equations \cite{Essler2020, Li2023, Teretenkov2024b}. An important question for further study is therefore which hierarchies of pure-state dynamics arise in such closed quantum systems. Another interesting direction is to clarify the connection between hierarchies of moments of creation and annihilation operators and matrix product operator approaches that have been actively developed in recent years \cite{Yashin2026}.

   \section{Acknowledgments}
   
   The author is sincerely grateful to Y.~A.~Biriukov, O.~V.~Lychkovskiy, M.~S.~Shustin and V.~I.~Yashin for fruitful discussions of the problems considered in the paper.
   
   \section{Funding}
   
   This work was supported by the Russian Science Foundation under grant no. 25-21-00700, https://rscf.ru/en/project/25-21-00700/.
   
   %
   % The Bibliography
   %

\end{document}